\documentclass[prx,aps,floatfix,amsmath,amssymb,superscriptaddress,tightenlines,twocolumn]{revtex4}
\usepackage{graphicx}
\usepackage{hyperref}
\usepackage{amsmath}
\usepackage{tikz}
\usepackage{soul,xcolor}
\setstcolor{red}

\newcommand{\Rom}[1]{\uppercase\expandafter{\romannumeral #1\relax}}
\usepackage{amsthm}
\usepackage[shortlabels]{enumitem}

\newtheorem*{assumption*}{\assumptionnumber}
\providecommand{\assumptionnumber}{}
\makeatletter

\newtheorem*{axiom*}{\assumptionnumber}
\providecommand{\assumptionnumber}{}
\makeatletter

\newtheorem*{conjecture*}{\assumptionnumber}
\providecommand{\assumptionnumber}{}


\theoremstyle{definition}

\theoremstyle{theorem}
\newtheorem{theorem}{Theorem}

\theoremstyle{corollary}

\theoremstyle{lemma} 
\newtheorem{lemma}[theorem]{Lemma}

\theoremstyle{Proposition} 
\newtheorem{Proposition}[theorem]{Proposition}

\theoremstyle{Conjecture}

\usepackage{amssymb}

\newcommand{\calC}{{\mathcal C}}

\newcommand{\Tr}{{\rm Tr}}

\definecolor{BSorange}{RGB}{255,127,0}

\begin{document}
    \title{Verlinde formula from entanglement} 
    
    \author{Bowen Shi}
    \affiliation{Department of Physics, The Ohio State University, Columbus, OH 43210, USA}

    \date{\today}
    
    \begin{abstract}
    We derive the Verlinde formula from a recently advocated set of axioms about entanglement entropy [B. Shi, K. Kato, I. H. Kim, arXiv:1906.09376 (2019)]. For any state that obeys these axioms, we can define a quantity that can be identified as the topological $S$-matrix of an abstract anyon theory. We show that the $S$-matrix is unitary and that it recovers the fusion multiplicities of the underlying anyon theory through the Verlinde formula. Importantly, we rigorously prove the modularity of the theory, which further implies that the mutual braiding statistics of anyons are nontrivial. The key to the proof is a generalized quantum state merging technique,  which generates a topology beyond that of any subsystem of the original physical system. 
    \end{abstract}

    \pacs{}
    \maketitle
    
    \section{Introduction}
    
    Interacting quantum many-body systems can exhibit a variety of exotic phenomena. In a strongly interacting regime, the low-energy excitations may obey emergent laws \cite{anderson1972more} that do not necessarily hold at the level of the constituent particles. In the context of two-dimensional (2D) gapped phases, anyons,
    which appear naturally in topologically ordered systems \cite{wen1990topological}, are expected to be described by the algebraic theory of anyon~\cite{2006AnPhy.321....2K}. It is a general framework that captures the fusion and braiding properties of anyons \footnote{The mathematical framework underlying the algebraic theory of anyon is the unitary modular tensor category (UMTC), see \cite{etingof2016tensor} for a review. It is sometimes also referred to as topological quantum field theory (TQFT), which may have different meanings (either that by Witten~\cite{witten1988topological} or Turaev~\cite{turaev1994quantum}), depending on the context.}.  
    In particular, it is expected that the theory is \emph{modular}, which is the requirement of a unitary topological $S$-matrix. Modularity is tied to the braiding nondegeneracy of the theory.
    Furthermore, there is a nontrivial relation between the $S$-matrix and the fusion multiplicities, which is known as the Verlinde formula.

    Historically, the fusion rules and the Verlinde formula were first derived in a different physical context, i.e., the general framework of conformal field theory \cite{belavin1984infinite,verlinde1988fusion,moore1989classical,francesco2012conformal}. The key underlying assumption is conformal invariance. This assumption is physically natural for a critical point in which scale invariance is expected to emerge. On the other hand, conformal invariance is not a physically natural assumption for gapped systems.
    
    In the physical context of 2D gapped systems, the algebraic theory of anyon~\cite{2006AnPhy.321....2K} is well known at this point. Despite its success, the properties of fusion multiplicities, the requirement of a unitary $S$-matrix, and the Verlinde formula, etc. are essentially plugged in from the underlying axioms of the theory. 
    It remains a fundamental problem to derive these axioms from an arguably more physical assumption for 2D gapped systems.
    
    One such attempt was recently made in Ref.~\cite{2019arXiv190609376S}. The authors identified two local entropic conditions (axiom {\bf A0} and {\bf A1} of \cite{2019arXiv190609376S}) as a reasonable starting point to derive the axioms of the anyon theory. What gives credence to these axioms is the conjectured area law of entanglement \cite{2006PhRvL..96k0404K,2006PhRvL..96k0405L}, which would imply the proposed axioms. The two axioms capture the quantum Markov chain structure of gapped 2D ground states \cite{2003RvMaP..15...79P,2004CMaPh.246..359H,2013PhRvL.111h0503K,2014arXiv1405.0137K,2015PhRvB..92k5139K,2016PhRvA..93b2317K,2019PhRvB..99c5112S,shi2019seeing}, which is a statement about the many-body quantum correlation. While currently there is no rigorous proof of the entanglement area law in 2D, it is wildly accepted at the point. It is explicitly verified in a large class of exactly solved models \cite{2003AnPhy.303....2K,2005PhRvB..71d5110L}, and it shows excellent agreement with numerical results  \cite{2011NatPh...7..772I,jiang2012spin}. It should be pointed out the axioms hold only approximately in realistic models, and there are fine-tuned 2D gapped states which violate {\bf A1} at all length scales \cite{Bravyi2008,2016PhRvB..94g5151Z,2019PhRvL.122n0506W}. Nevertheless, current evidence is still consistent with the conjecture that the area law is a good approximation on large length scales for a probable 2D gapped ground state.

    By starting from {\bf A0} and {\bf A1}, Ref.~\cite{2019arXiv190609376S} has defined the superselection sectors (i.e., anyon types), the fusion rules, and has derived the set of conditions that the fusion multiplicities are expected to satisfy. Furthermore, the authors independently derived the well-known formula of topological entanglement entropy (TEE) \cite{2006PhRvL..96k0404K,2006PhRvL..96k0405L}.
    These data are uniquely specified if one has access to a single quantum state. The superselection sectors, fusion multiplicities, and the consistency conditions are captured by the structure and the self-consistency relations of the information convex sets~\cite{2019arXiv190609376S}. It also shows that a deformable unitary string operator exists, which creates an anyon-antianyon pair. 
    
    In this work we show that a unitary $S$-matrix can be defined in the framework~\cite{2019arXiv190609376S}. We define a quantity which can be identified with the $S$-matrix, and we show it recovers the fusion multiplicities through the Verlinde formula. This implies that the theory is modular. Physically, this means entanglement area law implies the nontrivial braiding statistics of anyons in addition to the fusion rules. 
    
    We further expect the logic developed in this work to be useful in the classification of 3D topologically ordered systems, topological defects, and the gapped domain walls separating two gapped phases.

    \section{Background}
    Because our derivation is built upon the framework \cite{2019arXiv190609376S}, we first recall the setup and collect the relevant facts.
    We consider a 2D quantum many-body system, the Hilbert space of which has a tensor product structure. 
    We consider a quantum state $\vert\psi\rangle$ of this quantum system, which satisfies the following two conditions on each bounded-radius disk, see Fig.~\ref{Axioms}. Let $S_A=-\Tr(\sigma_A \ln \sigma_A)$ be the von Neumann entropy of the reduced density matrix $\sigma_A=\Tr_{\bar{A}} \vert\psi\rangle\langle \psi\vert$, where $\bar{A}$ is the complement of $A$. When the disk is divided into $BC$, we require that
    \begin{equation}
    S_{BC} + S_{C} - S_{B} = 0. \label{eq:A0}
    \end{equation}
    When the disk is divided into $BCD$, we require that
    \begin{equation}
    S_{BC} + S_{CD} - S_{B} - S_{D} = 0. \label{eq:A1}
    \end{equation}
    These two local  entropic conditions are known as axiom {\bf A0} and {\bf A1} in  \cite{2019arXiv190609376S}. Credit should be given to Kim for the original thoughts on these two conditions \footnote{These two entropic conditions were originally proposed in \cite{2014arXiv1405.0137K}. The first attempt at deriving the axioms of anyon theory from these conditions was presented in a conference Ref.~\cite{Kim2015sydney}.}.
    We shall refer to the state $\vert \psi\rangle $ (or its reduced density matrices) as the reference state.
    Physically interesting examples are the ground states of topologically ordered systems \cite{wen1990topological}, for which the bounded-radius disks and the subsystems $B$, $C$, $D$ are required to be larger than the correlation length. This approach is Hamiltonian independent.

    \begin{figure}[h]
	\centering
    \begin{tikzpicture}
    \node[] (E) at (0,0) {\includegraphics[scale=0.600]{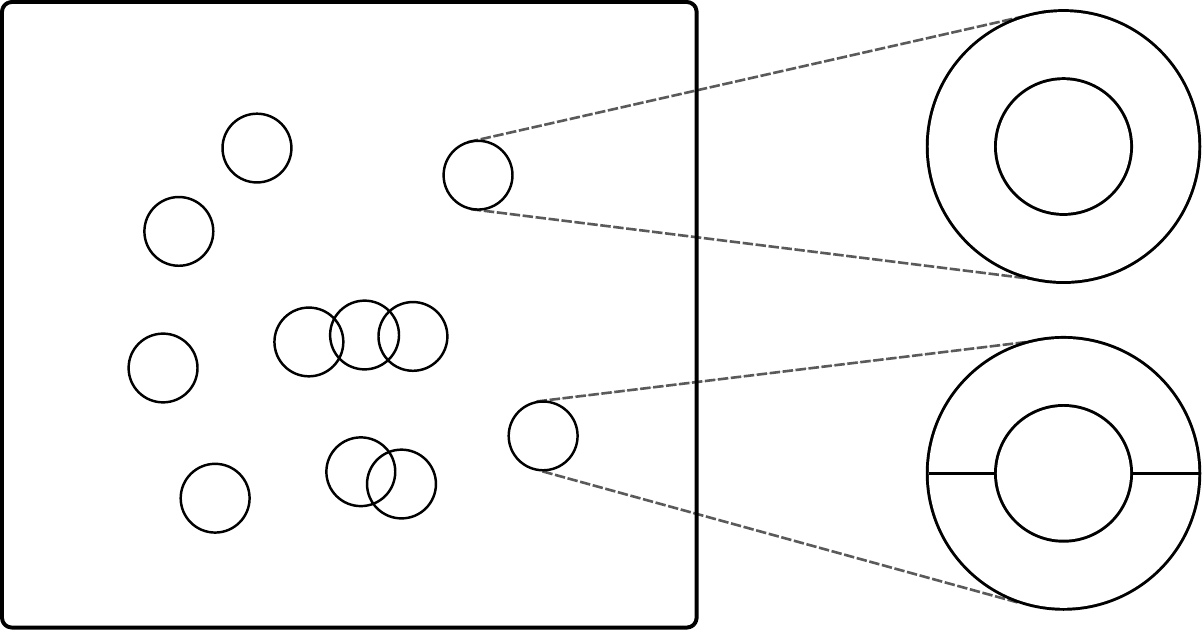}};
    \node[] (A) at (2.82,1.03) {\scriptsize{$C$}};
    \node[] (B) at (2.82,1.63) {\scriptsize{$B$}};
    \node[] (A) at (2.82,-0.37) {\scriptsize{$B$}};
    \node[] (B) at (2.82,-0.97) {\scriptsize{$C$}};
    \node[] (A) at (2.82,-1.57) {\scriptsize{$D$}};
    \end{tikzpicture}
	\caption{The reference state $\vert \psi\rangle$ on a 2D plane. It satisfies two entropic conditions (axioms {\bf A0} and {\bf A1} of Ref.~\cite{2019arXiv190609376S}), namely, for every bounded-radius disk, which is divided into $BC$ or $BCD$, we require Eq.~(\ref{eq:A0}) or Eq.~(\ref{eq:A1}) to hold. The subsystems $B$, $C$, $D$ can be deformed provided that the deformation keeps the topology intact.}
	\label{Axioms}
    \end{figure}
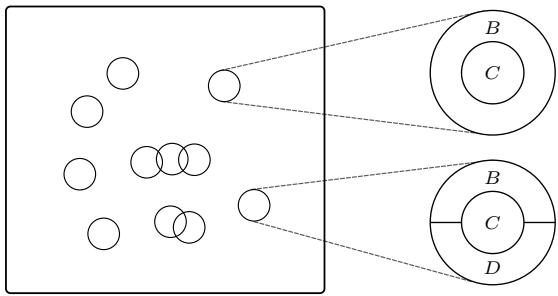

    Given a reference state satisfying axiom {\bf A0} and {\bf A1},  a finite set of superselection sector labels $\calC=\{1,a,b,\cdots \}$ and a set of fusion multiplicities $\{ N_{ab}^c \}$ can be defined. Here $1$ is the unique vacuum sector, and each $a\in \calC$ has a unique antiparticle $\bar{a}\in\calC$. The multiplicities, which are nonnegative integers responsible for the fusion rules $a\times b=\sum_c N_{ab}^c \,c$, are shown to satisfy all the expected conditions (see Appendix~\ref{ap:fusion_rules} for the conditions). The quantum dimensions $\{d_a\}$ can be uniquely defined according to $d_a d_b =\sum_c N_{ab}^c d_c$, and $\mathcal{D}=\sqrt{\sum_a d_a^2}$ is the total quantum dimension. 
    
    These universal data and the consistency relations emerge from the geometry and self-consistency relations of the information convex sets~\footnote{The definition in \cite{2019arXiv190609376S}, which is based on a single quantum state, was first considered in a slide of I. H. Kim, 2015 \cite{Kim2015sydney}.  The author introduced the terminology information convex (set) in \cite{2019PhRvB..99c5112S} without knowing Kim's slides, and the definition was based on a particular form of the Hamiltonian. Nonetheless, the author was inspired by a discussion in \cite{2015PhRvB..92k5139K}. The information-theoretic consistency of information convex sets was briefly investigated in \cite{shi2019seeing} based on structure assumptions.}. The information convex set $\Sigma(\Omega)$ is a convex set of density matrices defined for a subsystem $\Omega$, given the reference state $\vert\psi\rangle$. The sets are isomorphic for a pair of subsystems that can be smoothly deformed into each other, and every element is locally indistinguishable from the reference state. We will need part of the structure theorems of the information convex sets proved in Ref.~\cite{2019arXiv190609376S}. For an annulus $X$ (which is contained in a disk region), the information convex set $\Sigma(X)$ is a simplex with a finite set of extreme points. These extreme points are in one-to-one correspondence with the set of superselection sectors,  $\{\sigma^a_{X} \}_{a\in\calC}$. Distinct extreme points are orthogonal, i.e., $\sigma_X^a\perp \sigma_X^b$ for $a\ne b$. The reference state reaches the extreme point $\sigma^1_X$, which carries the vacuum sector. 
    
    Reference~\cite{2019arXiv190609376S} further derives the well-known formula of TEE, $\gamma=\ln \mathcal{D}$. This value comes from the entropy difference, $2\gamma = S(\sigma^{\ast}_X)- S(\sigma^1_X)$, where $\sigma^{\ast}_X$ is the maximal-entropy element in the ``center" of $\Sigma(X)$. All superselection sectors contribute to the TEE because they correspond to distinct extreme points 
    \footnote{Our setup and axioms are natural for models with bosonic local degrees of freedom. Fermionic models deserve a separate study.}. Moreover, the reference state $\vert \psi\rangle $ is long-range entangled \cite{2010PhRvB..82o5138C} if $\Sigma(X)$ has more than one extreme point.
    
    Finally, Ref.~\cite{2019arXiv190609376S} shows the existence of a deformable unitary string operator which creates a pair $(a,\bar{a})$. The positions of the anyons can be chosen to be two bounded-radius disks. On an annulus $X$ surrounding $a$, the extreme point $\sigma^a_X\in \Sigma(X)$ is reached.
    The support of the string can be deformed freely in a topological manner.

    \section{The main result and its proof}
    
    The main result of this work is the definition of a quantity for a reference state $\vert\psi\rangle$, which is identified with the topological $S$-matrix of the underlying anyon theory. We show the $S$-matrix we define is unitary, and it recovers the fusion multiplicities through the Verlinde formula
    \begin{equation}
    N_{ab}^c =\sum_{x\in \calC}\frac{S_{ax} S_{bx} S_{\bar{c} x} }{ S_{1x}}, \label{eq:Verlinde}
    \end{equation}
    where the components of the $S$-matrix, $S_{ab}$ with $a,b\in \calC$, have $S_{a1}= \frac{d_a }{ \mathcal{D}}$ and  the following symmetries:
    \begin{equation}
    S_{ab}= S_{ba},\quad S_{ab} = S^{\ast}_{ \bar{a} b}. \label{eq:S_sym}
    \end{equation}
    This establishes the \emph{modularity} of the theory, and it is tied to the \emph{braiding nondegeneracy}. It corresponds to an independent axiom of the algebraic theory of anyon~\footnote{It does not follow from the pentagon equation or the hexagon equations.}.
    We derive this result from axioms {\bf A0} and {\bf A1} of Ref.~\cite{2019arXiv190609376S}.

    \subsection{Our definition of the $S$-matrix}

    \begin{figure}[h]
	\centering
        \begin{tikzpicture}
\node[] (E) at (0,0) {\includegraphics[scale=0.66]{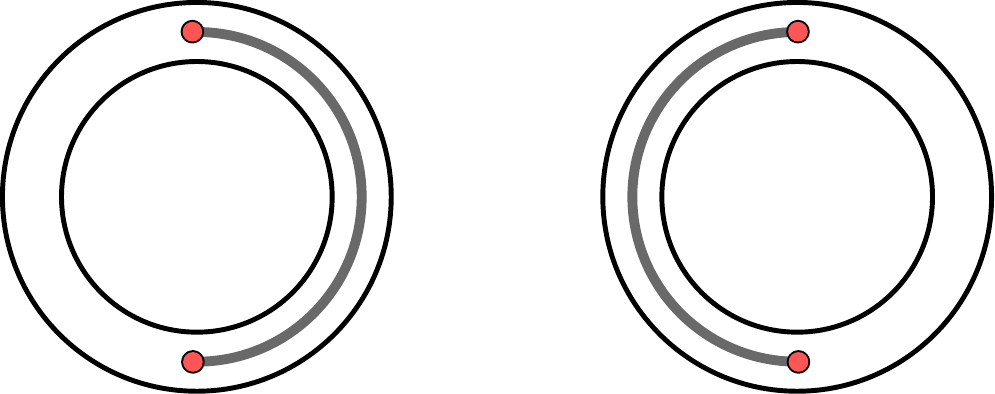}};
\node[] (A) at (-2.277,1.1) {{$\bar{a}$}};
\node[] (B) at (-2.277,-1.1) {{${a}$}};
\node[] (A) at (2.277,1.1) {{$\bar{a}$}};
\node[] (B) at (2.277,-1.1) {{${a}$}};
\node[] (A) at (-2.046,-1.75) {(a)};
\node[] (B) at (2.046,-1.75) {(b)};
\end{tikzpicture}
	\caption{An annulus $X$ and string operators supported within it. (a) String operator $U_R^a$ which creates a pair of excitations $a$ and $\bar{a}$ on the reference state.  (b) The string operator $U_L^a$ is obtained by deforming $U_R^a$  on the reference state.}
	\label{ULUR_PDF}
\end{figure}
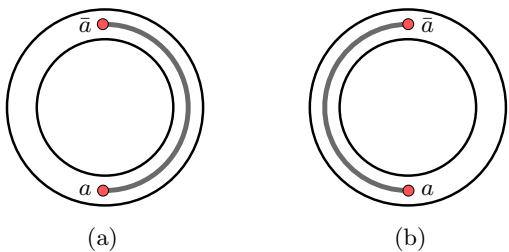
    
    We define $S_{ab}$ as follows:
    \begin{eqnarray}
    S_{ab}&\equiv& \frac{d_a d_b}{ \mathcal{D}} f_{ab}, \label{def:S}\\
    f_{ab} &\equiv& \Tr (U_L^{a\dagger} U_R^a \sigma^b_X). \label{def:fab}
    \end{eqnarray}
    Here $X$ is an annulus, and $U_R^a$ is an operator which creates a pair of anyons $(a,\bar{a})$, see Fig.~\ref{ULUR_PDF}. $U_L^a$ is obtained from $U_R^a$ by a deformation on the reference state, namely, we require that $U_L^a\vert \psi\rangle = U_R^a \vert \psi\rangle$.
    $\sigma_X^b$ is an extreme point of the information convex set $\Sigma(X)$.
    By definition, $\Tr (U_L^{a\dagger} U_R^a \sigma^1_X)=1$. This implies that $f_{a1}=1$, $\forall a$.

    This $S_{ab}$ is well defined in the sense that it is invariant under the deformation of three things: the annulus $X$,  the support of the strings, the positions of the anyons. To establish this fact, we first notice that the extreme point $\sigma^b_X$ can be obtained by acting string operators on the reference state. Therefore, one can rewrite Eq.~(\ref{def:fab}) as an expectation value of four string operators. First, for generic deformable unitary strings (see Fig.~\ref{UV}), we define
    \begin{equation}
    f(U,V)\equiv \langle \psi\vert U_L^{\dagger} V_R^{\dagger} U_R V_L \vert \psi\rangle. \label{eq:exp_uv}
    \end{equation}
    It recovers $f_{ab}$ when the strings carry fixed sectors, i.e.,
    \begin{equation}
    f_{ab}= f(U^a,V^b)=\langle \psi\vert U_L^{a\dagger} V_R^{b\dagger} U_R^{a} V_L^{b} \vert \psi\rangle.  \nonumber
    \end{equation}
    Because these string operators act directly on the reference state (either to the left on $\langle \psi\vert$ or to the right on $\vert\psi\rangle$), small deformation of any one of them will leave $f(U,V)$ invariant. Moreover, modifying the string operators by a slight change of the position of an excitation (without passing the excitation through another string) will not affect the value of $f(U,V)$. This is because applying this new string operator on the reference state is equivalent to applying the original one and then applying an additional unitary operator supported on the union of two disk-like regions. In the expectation value (\ref{eq:exp_uv}), these additional unitary operators are canceled.

    \begin{figure}[h]
	\centering
    \begin{tikzpicture}
    \node[] (E) at (0,0) {\includegraphics[scale=0.53]{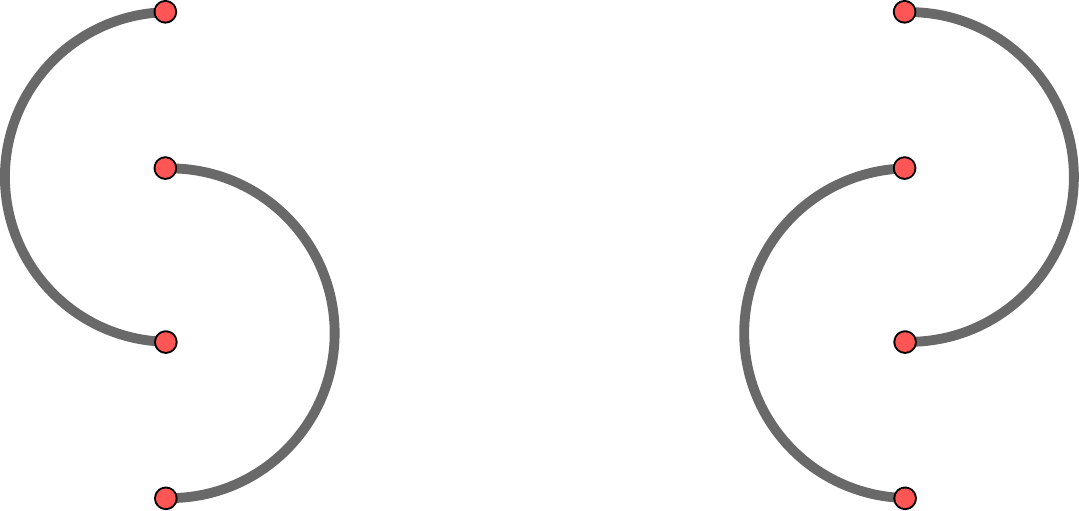}};
    \node[] (A) at (-2.025,-1.85) {\small{(a)}};
    \node[] (B) at (1.975,-1.85) {\small{(b)}};
    \node[] (B) at (0.8,-0.4105 ) {{$U_L$}};
    \node[] (B) at (3.15,0.4105) {{$V_R$}};
    \node[] (B) at (-0.8,-0.4105 ) {{$U_R$}};
    \node[] (B) at (-3.15,0.4105) {{$V_L$}};
    \end{tikzpicture}
	\caption{Two distinct ways to create four excitations: (a) with $U_R$ and $V_L$, (b) with $U_L$ and $V_R$. Here $U_L \vert \psi\rangle = U_R\vert \psi\rangle$ and  $V_L \vert \psi\rangle = V_R\vert \psi\rangle$. Depending on the context of the discussion, an operator may either correspond to a string carrying a fixed sector or a string bundle.}
	\label{UV}
    \end{figure}
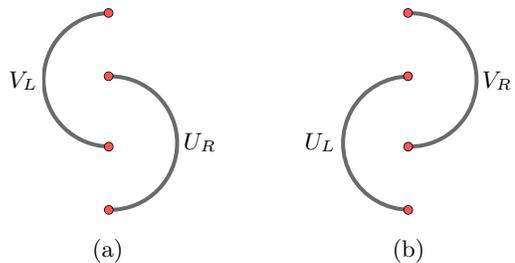

    Using the trick of deforming the string operators, taking a partial trace, and making use of the aforementioned invariant property, one finds 
    \begin{equation}
    f_{ab} = f_{ba},\quad
    f_{ab} = f^{\ast}_{\bar{a}b}. \label{eq:sym_f}
    \end{equation}
    In more detail, to verify these identities, one can diagrammatically represent both sides of the identity and then smoothly deform one to another. The deformation involves both the strings and the anyon positions.
    These identities imply that our definition of $S$-matrix obeys the requisite symmetries~(\ref{eq:S_sym}).

    \subsection{Proof of the Verlinde formula}
    
    To facilitate the proof, we remark on the approach of deriving (\ref{eq:sym_f}). First, the deformation of string operators and taking a partial trace allows us to obtain a quantity in a few different ways. By matching these results, one can derive a constraint. Second, the deformation of a string operator is a rather general property. It works not only for a string which carries a fixed sector but also for a \emph{string bundle}, which is a product of multiple string operators with disjoint supports (see Fig.~\ref{string_bundle})  \footnote{The same trick works for a   string which prepares a generic element of $\Sigma(X)$ for an annulus $X$ surrounding an excitation. This fact is not needed in the proof.}.
    By applying the above idea to string bundles, we obtain the following proposition:

\begin{figure}[h]
	\centering
    \begin{tikzpicture}
    \node[] (E) at (0,0) {\includegraphics[scale=0.7]{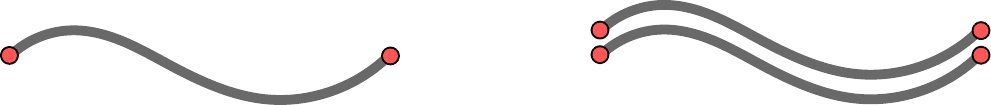}};
    \node[] (B) at (-2.1,-0.9) {(a)};
    \node[] (B) at (2.1,-0.9) {(b)};
    \end{tikzpicture}
	\caption{(a) A single  string. (b) A string bundle. In this particular figure, the string bundle consists of two strings.}
	\label{string_bundle}
\end{figure}

    \begin{Proposition} \label{Prop:1}
        The $S$-matrix  we define satisfies 
        \begin{equation}
        \sum_c N_{ab}^c S_{cx} = \frac{ S_{ax} S_{bx} }{ S_{1x} }. \label{eq:NS1}
        \end{equation}
    \end{Proposition}
    
    \begin{proof}
        We show that $\sum_c P_{(a\times b\to c) } f_{cx} = f_{ax} f_{bx}$, then Eq.~(\ref{eq:NS1}) follows. Here $P_{(a\times b\to c)} =\frac{N_{ab}^c d_c}{d_a d_b}$. Let us consider $f(U^{ab},V^{x})$, where $U^{ab}$ is a string bundle consisting of two strings with sectors $a$ and $b$, and $V^x$ is a string with sector $x$. We calculate $f(U^{ab},V^x)$ in two ways. 
        
        First, we have
        \begin{equation}
        \begin{aligned}
        f(U^{ab},V^x)&=\Tr(U^{ab\dagger}_L U_R^{ab}\sigma_X^x)\\
        &=f_{ax} f_{bx}. \label{eq:fuv_1}
        \end{aligned}
        \end{equation}
        In the first line, we have deformed $V_R^{x\dagger}$ and done a partial trace such that the remaining subsystem $X$ is an annulus containing $U^{ab}_L$ and $U^{ab}_R$.
        In the second line, we have used the fact that the extreme point $\sigma^x_X$ is ``factorizable." Specifically, in Ref.~\cite{2019arXiv190609376S}, it was shown that the extreme points of an annulus, restricted to disjoint sub-annuli, have a tensor product form. Therefore, we can split $U_{L}^{ab}$ and $U_{R}^{ab}$ into two families, $(U_L^{a}, U_R^{a})$ and $(\widetilde{U}_{L}^b, \widetilde{U}_R^b)$, so that the string operators for $a$ and $b$ are supported on disjoint sub-annuli. Because operators belonging to different families commute with each other, the expectation value becomes:
        \begin{equation}
        \begin{aligned}
            \Tr(U^{ab\dagger}_L U_R^{ab}\sigma_X^x)&= \Tr(U_{L}^{a\dagger}U_R^a \widetilde{U}_L^{b\dagger}\widetilde{U}_R^b \sigma_X^x)\\
            &=\Tr(U_{L}^{a\dagger}U_R^a \sigma_{X_1}^x) \Tr(\widetilde{U}_L^{b\dagger} \widetilde{U}_R^b \sigma_{X_2}^x), \nonumber
        \end{aligned}
        \end{equation}
        where $X_1,X_2 \subseteq X$ are disjoint sub-annuli of $X$. 
        
        Second, we deform the string $U_R^{ab}$ in the same manner. Because the string bundle $U^{ab}$ can produce sector $c$ on an annulus surrounding $a$ and $b$ with probability $\frac{N_{ab}^c d_c}{d_a d_b}$ \cite{2019arXiv190609376S,shi2019seeing}, we obtain the following expression:
        \begin{equation}
        f(U^{ab},V^x)= \sum_{c} P_{(a\times b\to c) } f_{cx}. \label{eq:fuv_2}
        \end{equation} 
        
        By matching the two expressions (\ref{eq:fuv_1}) and (\ref{eq:fuv_2}) one obtains Eq.~(\ref{eq:NS1}).
    \end{proof}

   Note that Proposition~\ref{Prop:1} in itself does not imply modularity (i.e., that the $S$-matrix is unitary). For example, a solution like $f_{ab}=1$, $\forall a,b$ is consistent with Eq.~(\ref{eq:NS1}) but it leads to a noninvertible $S$-matrix. We need a concrete statement on the nontrivial braiding. The key is the following lemma:
    
    \begin{lemma} \label{lemma_1}
        Let $\sigma^{\ast}_{X} =\sum_{a} \frac{d_a^2}{\mathcal{D}^2} \sigma^a_X$ be the maximal-entropy element of $\Sigma(X)$, then
        \begin{equation}
        \Tr (U_L^{a\dagger} U_R^a \sigma^{\ast}_{X}) =\delta_{a,1}. \label{eq:lemma}
        \end{equation}
    \end{lemma}
    See Sec.~\ref{Sec.Lemma_proof} for the proof Lemma~\ref{lemma_1}.
    Based on Lemma~\ref{lemma_1}, we show that the $S$-matrix is unitary, and we further derive the Verlinde formula. 
    
    \begin{Proposition}
        The    $S$-matrix is unitary 
        and the Verlinde formula (\ref{eq:Verlinde}) holds.
    \end{Proposition}
    \begin{proof}
        We only need to show that the $S$-matrix is unitary. The Verlinde formula follows from unitarity and Eq.~(\ref{eq:NS1}). Equation~(\ref{eq:lemma}) implies that $\sum_x S_{1x} S_{ax}= \delta_{a,1}$. Multiplying $S_{1x}$ to both sides of  Eq.~(\ref{eq:NS1}), doing the sum of $x$, and using $N_{ab}^1=\delta_{b,\bar{a}}$, one derives that 
        $\sum_x S_{ax} S_{bx} = \delta_{b,\bar{a}}$.
        This, together with the symmetry properties (\ref{eq:S_sym}), implies that the $S$-matrix is unitary. This completes the proof.
    \end{proof}
    The same logic applies to a generic string bundle, and the end result is the Verlinde formula for a generic number of excitations.

    \subsection{Proof of Lemma~\ref{lemma_1}}\label{Sec.Lemma_proof}

    \begin{figure}[h]
	\centering
    \begin{tikzpicture}
    \node[] (E) at (0,0) {\includegraphics[scale=0.6]{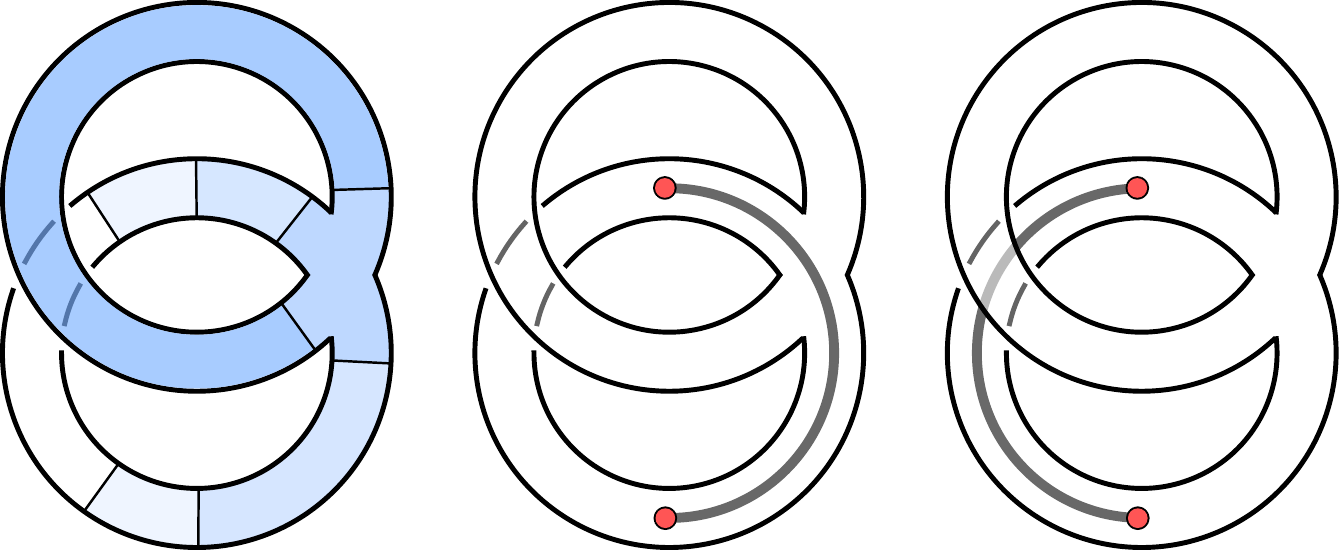}};
    \node[] (A) at (-2.87,1.47) {\scriptsize{$A$}};
    \node[] (B) at (-2.02,0) {\scriptsize{$B$}};
    \node[] (A) at (-2.57,0.48){\scriptsize{$C_1$}};
    \node[] (A) at (-3.15,0.48) {\scriptsize{$C_2$}};
    \node[] (A) at (-2.3,-1.3){\scriptsize{$C_1$}};
    \node[] (A) at (-3.15,-1.45) {\scriptsize{$C_2$}};    
    \node[] (B) at (-3.87,-0.65) {\scriptsize{$D$}};
    \node[] (A) at (-0.21, 0.51) {\footnotesize{$\bar{a}$}};    
    \node[] (B) at (-0.21,-1.485) {\footnotesize{$a$}};
    \node[] (A) at (3.05, 0.51) {\footnotesize{$\bar{a}$}};    
    \node[] (B) at (3.05,-1.485) {\footnotesize{$a$}};
    \node[] (B) at (-2.88,-2.07) {(a)};
    \node[] (B) at (-0.02,-2.07) {(b)};
    \node[] (B) at (2.86,-2.07) {(c)};
    \end{tikzpicture}
	\caption{(a) The merging of $\sigma^1_{ABC}$ and $\sigma_{CD}$, where $C=C_1C_2$. $ABCD$ is not a subsystem of the original system, and it has a topology equivalent to a torus with one hole.
		(b) The unitary string operator  $U_R^a$ is supported on $BC$. (c) The unitary string operator $U_L^a$ is supported on $CD$.  It is obtained from the deformation of $U_R^a$. }
	\label{Merging_PDF}
\end{figure}
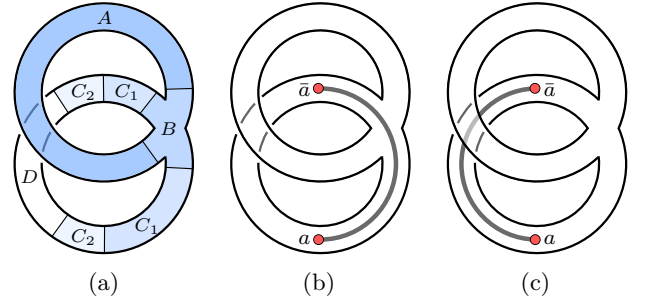

    \begin{proof}
        For $a=1$, Eq.~(\ref{eq:lemma}) is trivially true. In order to derive Eq.~(\ref{eq:lemma}) for the case of $a\ne 1$, we consider the merging process described in Fig.~\ref{Merging_PDF}(a). We merge 
        \footnote{ Merging is a process to obtain a unique density matrix from two density matrices supported on smaller regions. It works for quantum Markov states satisfying a few simple conditions, see \cite{2019arXiv190609376S,2016PhRvA..93b2317K}.
        	The merging technique is first introduced in \cite{2016PhRvA..93b2317K}. 
        	The fact that the merging result of a pair of elements of information convex sets is an element of an information convex set, which is necessary for the proof in this paper, is established in \cite{2019arXiv190609376S}.} 
        two reduced density matrices of the reference state, namely, $\sigma^1_{ABC}$ and $\sigma_{CD}$, where $C=C_1C_2$. We call the density matrix obtained by this merging as $\tau_{ABCD}$. (The two states can be merged because they are identical on $C$ and the conditional mutual information $I(AB:C_2\vert C_1)= I(C_1:D\vert C_2)=0$ for the reference state.)
        Note that while $ABC$ and $CD$ are subsystems of the original physical system, the support of the merged state is not. This is because $A$ and $D$ on the original physical system overlap nontrivially, yet in the merged state, they do not share any common region. For the state $\tau_{ABCD}$, $A$ and $D$ belong to different Hilbert spaces. What is important here is that the merged state $\tau_{ABCD}$ exists even though one cannot obtain such a state by tracing out subsystems from the original physical system. 
        
        Let us consider the reduced density matrices of $\tau_{ABCD}$ on annuli $ABC$ and $BCD$.
        $\Tr_{D} \tau_{ABCD} = \sigma^1_{ABC}$ carries the vacuum sector. 
        $\Tr_{A} \tau_{ABCD} = \sigma^{\ast}_{BCD}$ is the maximal-entropy element of  $\Sigma(BCD)$.
        After applying $U_R^a$ or $U_L^a$  onto $\tau_{ABCD}$, see Fig.~\ref{Merging_PDF}, the sectors seen on $AB$ are
        \begin{eqnarray}
        \Tr_{CD} (U_L^a \tau_{ABCD} U_L^{a\dagger}) &=& \sigma^1_{AB}, \nonumber \\
        \Tr_{CD} (U_R^a \tau_{ABCD} U_R^{a\dagger})&=&\sigma^{\bar{a}}_{AB}. \nonumber
        \end{eqnarray}
        Thus, the two density matrices $U_L^a \tau_{ABCD} U_L^{a\dagger}$ and $U_R^a \tau_{ABCD} U_R^{a\dagger}$ are orthogonal for $a\ne 1$. This fact follows from that $ \sigma^1_{AB}\perp  \sigma^{\bar{a}}_{AB}$ for $a\ne 1$ and the monotonicity of fidelity. Therefore,
        \begin{equation}
        \Tr(U_L^{a \dagger} U_R^a \sigma^{\ast}_{BCD})= \Tr(U_L^{a \dagger} U_R^a \tau_{ABCD})=0,\quad \forall\,a\ne 1. \nonumber
        \end{equation}
        Since $BCD$ can be any annulus, Lemma~\ref{lemma_1} is justified.
    \end{proof}

    We would like to remark on a counterintuitive aspect of the proof.
    A careful reader may imagine another quantum state, $\rho_{ABCD}$, which reduces to $\sigma^1_{ABC}$ and $\sigma^1_{BCD}$. Then, by applying the same logic, one seems to get a contradiction, namely, $\Tr(U_L^{a\dagger} U_R^a \sigma^1_{BCD})=0$. It is not a real contradiction. Instead, it means that such a state $\rho_{ABCD}$  cannot exist. 
    This phenomenon can be understood from the entropic ``uncertainty principle" in Refs.~\cite{2012PhRvB..85w5151Z,Jian2015}, which implies that annuli $ABC$ and $BCD$ cannot both obtain the vacuum sector. Note that the $ABCD$ in Fig.~\ref{Merging_PDF}(a) is topologically equivalent to the 1-hole torus considered in~\cite{2012PhRvB..85w5151Z}. In comparison, our method does not make use of the global topology of the system, and $\tau_{ABCD}$ is constructed given the reduced density matrices within a disk region. It makes our method applicable to a broader context, e.g., a sphere or a torus with ground-state degeneracy modified by a closed defect line.

    As a corollary, the $S$-matrix is encoded in a single quantum many-body state, and moreover, we only need the reduced density matrix within a disk region. This result should be contrasted with~\cite{2012PhRvB..85w5151Z}, which makes use of multiple ground states on a torus. 
    We would like to compare our result with another recent attempt~\cite{haah2016invariant} to define the $S$-matrix from one single ground state. It makes assumptions concerning the Hamiltonian and the operator algebra. As the author remarks, the method therein requires the assumption of a unitary modular tensor category description to complete the argument that the invariant constructed matches the $S$-matrix.
    In comparison, with axioms {\bf A0} and {\bf A1} of \cite{2019arXiv190609376S}, we are able to define the $S$-matrix and derive the Verlinde formula it obeys.

    \section{Summary and discussion}
    We have derived the Verlinde formula from a law of entanglement natural for 2D gapped systems (axioms {\bf A0} and {\bf A1} of the framework \cite{2019arXiv190609376S}). From a 2D quantum state satisfying these axioms, we define a unitary $S$-matrix, which recovers the fusion multiplicities through the Verlinde formula.
    It shows that axiom {\bf A0} and {\bf A1} imply the nontrivial mutual braiding statistics of anyons in addition to the previously identified fusion rules. It deserves a further study on whether the entire emergent physical law of anyons is implied by these two conditions.
    
    Both the fusion rules and the $S$-matrix are encoded in a single quantum state. It supports the conjecture that the entire set of universal data of a topologically ordered system is encoded in one single ground-state wave function. To justify this conjecture, one may further attempt to extract the topological spins. It is recently noticed that  $S$ and $T$ matrices do not completely determine an anyon model \cite{2017arXiv170802796M,2018arXiv180505736B,2018arXiv180603158K,2019arXiv190810381W}, and therefore additional topological invariants need to be considered. We have generalized the merging technique to produce a quantum state supported on a topology beyond that of any subsystem (e.g., Fig.~\ref{Merging_PDF}). Moreover, the definition of the information convex set naturally generalizes into this context. We expect this observation to be useful in future studies.
    One may further attempt to define $F$ and $R$ symbols from a state satisfying axioms {\bf A0} and {\bf A1}. In light of the recent  operational  definition of $F$ and $R$ symbols for microscopic models \cite{2019arXiv191011353K}, it is plausible that progress can be made. This is because the framework  \cite{2019arXiv190609376S} provides well-defined unitary processes.

    Finally, it should be emphasized that deriving the axioms of the algebraic theory of anyon is stronger than extracting the anyon data. The power of our method precisely lies in the fact that it can \emph{derive} the emergent laws.
    Even though the algebraic theory of anyon is well-known by now, there is plenty of space for further exploration. Namely, there are physical systems for which the abstract framework (analogous to the algebraic theory of anyon) is difficult to guess, but the analogy of axiom {\bf A0} and {\bf A1} can be easily inferred.
    Such examples include a large class of three-dimensional gapped phases, topological defects, and the gapped domain walls separating two gapped phases. The logic developed in Ref.~\cite{2019arXiv190609376S} and this work will be a powerful tool in the study of these systems.

    \section*{Acknowledgement}
    I am grateful to Isaac H. Kim for reading the draft and providing helpful feedback. I thank Jeongwan Haah for useful comments.
     I further thank Kohtaro Kato and Yuan-Ming Lu for things I learned from them previously.   This work is supported by the National Science Foundation under Grant No. NSF DMR-1653769.

    \appendix
    \section{Fusion rules}\label{ap:fusion_rules}
    The fusion multiplicities obey the following rules:
    \begin{enumerate}
        \item $N_{ab}^c = N_{ba}^c$.
        \item $N_{1a}^{c}  =\delta_{a,c}$.
        \item $N_{a b}^1 = \delta_{b,\bar{a}}$.
        \item $N_{a b}^{c} =N_{\bar{b}\bar{a}}^{\bar{c}}$.
        \item $\sum_i N_{a i}^{d} N_{b c}^i = \sum_j N_{a b}^j N_{jc}^{d}$.
    \end{enumerate}
    These rules are derived from axioms {\bf A0} and {\bf A1} in Ref.~\cite{2019arXiv190609376S}.

    \bibliography{ref}
    \bibliographystyle{apsrev}
\end{document}